\documentclass[letterpaper]{IEEEtran}

\usepackage{eulervm}

\usepackage[english]{babel}
\usepackage{amsmath,amssymb}
\usepackage{times}
\usepackage{relsize}
\usepackage{graphicx}
\usepackage{url}
\usepackage{booktabs}
\usepackage{multirow}
\usepackage{mparhack}
\usepackage{subfigure}
\usepackage{authblk}
\usepackage{amsthm}

\usepackage[noend]{algorithmic}
\usepackage[linesnumbered,ruled,vlined]{algorithm2e}

\usepackage{array}
\usepackage{cite}
\usepackage{eqparbox}
\usepackage{mdwmath}

\usepackage{epsfig}
\usepackage{xcolor}

\usepackage{mathtools}
\usepackage{bm}
\usepackage{bbold}
\usepackage[all]{xy}
\usepackage{etoolbox}
\DeclareMathAlphabet{\mathantt}{OT1}{antt}{li}{it}
\DeclareMathAlphabet{\mathpzc}{OT1}{pzc}{m}{it}

\usepackage{accents}

\newtheorem{theorem}{Theorem}
\newtheorem{definition}{Definition}

\newcommand\T{\rule{0pt}{2.4ex}}
\newcommand\B{\rule[-1.0ex]{0pt}{0pt}}

\DeclareFontFamily{OT1}{pzc}{}
\DeclareFontShape{OT1}{pzc}{m}{it}%
  {<-> s * [1.1] pzcmi7t}{}
\DeclareMathAlphabet{\mathpzc}{OT1}{pzc}%
                     {m}{it}

\def\B{\mathcal{B}}
\def\R{\mathcal{R}}
\def\T{\mathcal{T}}
\def\C{\mathcal{C}}
\def\G{\mathcal{G}}

\DeclareMathOperator{\argmin}{\arg\min}

\makeatletter
\patchcmd{\@maketitle}
  {\addvspace{0.5\baselineskip}\egroup}
  {\addvspace{-1.45\baselineskip}\egroup}
  {}
  {}
\makeatother

\title{\vspace{-0.4cm}Energy-Aware Wireless Relay Selection in Load-Coupled OFDMA Cellular
Networks}
\author[1]{\vspace{-0.3cm}Lei You}
\author[1,2]{Di Yuan}
\author[1]{Nikolaos Pappas}
\author[1]{Peter V\"{a}rbrand}
\affil[1]{{\small Department of Science and Technology, Link{\"o}ping
University, Sweden}}
\affil[2]{{\small Department of Information Technology, Uppsala University,
Sweden}}
\affil[ ]{
    \small\texttt{ \{lei.you; di.yuan; nikolaos.pappas@liu.se\}
    \{petva@itn.liu.se\}}
\/
\vspace{-0.3cm}
}

\begin{document}
\maketitle

\begin{abstract}
We investigate transmission energy minimization via optimizing wireless relay
selection in orthogonal-frequency-division multiple access (OFDMA) networks. We
take into account the impact of the load of cells on transmission energy.  We
prove the $\mathcal{NP}$-hardness of the energy-aware wireless relay selection
problem.  To tackle the computational complexity, a partial optimality condition
is derived for providing insights in respect of designing an effective and
efficient algorithm.  Numerical results show that the resulting algorithm
achieves high energy performance.
\end{abstract}

\vspace{-0.5cm}
\section{Introduction}

Relay techniques provide coverage extension, alleviate fading effects in wireless
channels, and lead to more rapid network roll-out to improve the overall system
energy efficiency~\cite{Andrews:ez,Michalopoulos:2015bi,Sheng:2015ks}. 
In meeting the fast growing demand of mobile
communication and the increase of user density, wireless relaying is
viewed as a promising technique for the upcoming 5G
\cite{EricssonABErikDahlman:2014uja}.
It is
shown that wireless backhaul technologies have competitive advantages over the
fiber-based solution~\cite{dohler20165g}. 
In 5G, outdoor relays are likely to be densely deployed in urban
areas, which may cause the cost of installing fiber-based relay nodes 
to reach an unacceptable level. For the indoor scenarios, wireless backhauling may
provide better flexibility and cost-efficiency, compared to a fiber-based
solution~\cite{release13-1}. In addition, though fiber-based
backhauling has advantage in capacity, reliability, and robustness for transmission, there
are cases in which wired backhauling is impossible (e.g.~short-term links
for emergency/disaster relief), hence making the wireless solution to be the
only option for such scenarios~\cite{dohler20165g}. 

There are two types of relaying modes in terms of wireless
backhauling~\cite{release13-1,Gora:2011bc}. One is called ``out-band'' mode, in
which the backhaul and access links operate on different carriers.  The other is
``in-band'' mode, meaning that there is no explicit splitting in frequency
resource between backhaul links and access links~\cite{release13-1,Gora:2011bc}.
Compared to the former, the latter does not require a pre-defined separation in
the frequency domain. Moreover, if relays are required to operate on a single
carrier, then there is no possibility to make separation for implementing
out-band relay mode, and thus in-band relay would be the only option in this
case~\cite{Gora:2011bc}.

Recently, studies~\cite{Cavalcante:2014jd,Ho:2015hw} investigated energy
minimization in orthogonal-frequency-division multiple access (OFDMA) networks,
under an interference model proposed in~\cite{Siomina:eq}. This model
characterizes the coupling relationship among the load of cells, which is
defined to be the proportion of consumed time-frequency resource in each cell.
The model is therefore named as a ``\textit{load-coupling}''
model~\cite{Siomina:eq}.  However, understanding and analyzing load coupling for
relays with wireless backhauling is not straightforward.  In this paper, we
provide significant extensions of the model to wireless relay scenarios, following the
LTE-advanced standard of wireless relays in\cite{release13-1}. We formulate the
energy-aware relay selection problem, named \textit{MinE}, and prove its
computational hardness.  Moreover, we derive an optimality condition, based on
which a relay selection algorithm is proposed for solving \textit{MinE}.
Numerical results show significant improvement on network energy consumption,
compared to the standard strategy of strongest-cell association.

\vspace{-0.3cm}
\section{System Model}
\label{sec:sys}

\subsection{Network Model}
We consider a heterogeneous cellular network (HetNet) with macro cells (MCs),
user equipments (UEs), and relay cells (RCs). Denote by
$\B=\{1,2,\ldots,n_{\B}\}$ the set of MCs, $\T=\{1,2,\ldots,n_{\T}\}$
the set of UEs, and $\R=\{n_{\T}+1,n_{\T}+2,\ldots,n_{\T}+n_{\R}\}$ the
set of RCs. We focus on downlink transmission in this paper.  For any UE
$j\in\T$, the set of $j$'s candidate serving cells is denoted by $\C_j$. For any
RC $k\in\R$, denote by $\C_k$ the set of $k$'s candidate MCs for establishing
the backhaul link.  The relay selection aims at 1) choosing a serving
cell out of $\C_j$ for all $j\in\T$, and 2) finding for each RC $k\in\R$ an
MC out of $\C_k$ to establish the backhaul link, so as to minimize the
network transmission energy.

We assume in-band wireless relay transmission~\cite{release13}, which implies no
explicit splitting of available time-frequency resource between the backhaul
links and the access links. To avoid the loop interference~\cite{Gora:2011bc},
the backhaul and access links should operate on orthogonal resources, meaning
that, within the area of each MC, the time-frequency resource units (RUs)
utilized by the two types of links do not overlap. Thus some of the links
preserve orthogonality with each other. We refer to 
\figurename~\ref{fig:illustration} for an illustration.
\begin{figure}[!h]
\vskip -10pt
    \centering
    \includegraphics[width=0.55\linewidth]{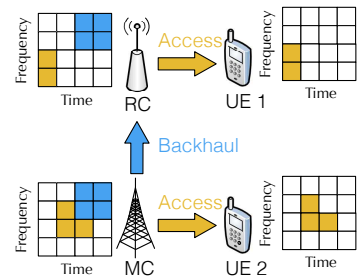}
    \vskip -10pt
    \caption{Illustration of the HetNet model. In this example, there are one
        MC, one RC, and two UEs. UE 1 and UE 2 are served by the RC and MC,
        respectively. The RC is doing a wireless backhauling with the MC\@. The
        used resource for access links and backhaul links is marked by yellow
    and blue colors, respectively. In time and frequency, the resource allocated to
backhauling has orthogonality to the resource used by the access links.}
\label{fig:illustration}
\vspace{-0.3cm}
\end{figure}

In the remaining parts of this paper, we use the term \textit{``orthogonal
links''} to refer to those links using orthogonal time-frequency resource for
transmission. Such links are said to be orthogonal to each other. Below, we
discuss the characterization of orthogonality. For the sake of presentation, consider given
association of UE access and RC backhauling, and denote by $\R_i$ the set of RCs
with a wireless backhaul connected to MC $i$ for each $i\in\B$. The set of UEs
served by any cell (MC or RC) $i$ is represented by $\T_i$.  Denote by tuple
$\langle i,j\rangle$ any (backhaul or access) link from $i$ to $j$. For any
access link $\langle i,j \rangle$ with $i\in\B$ and $j\in\T_i$, denote by
$\mathcal{L}_{ij}=\{\langle i,v\rangle:v\in\R_i\cup\T_i\}$ the set of links that
preserve orthogonality to link $\langle i,j\rangle$.  For $k\in\R$, suppose that
it is connected with some MC $i$ with a backhaul link.  We define
$\mathcal{L}_{kj}=\{\langle k,v\rangle:v\in\T_k\}\cup\{\langle i,k \rangle\}$ to
be the set of links having the intra-cell orthogonality to the access link
$\langle k,j\rangle$ with $j\in\T_k$.  And for the backhaul link $\langle i,k
\rangle$, we define $\mathcal{L}_{ik}=\{\langle i,v\rangle:v\in\R_i\cup\T_i\}
\cup\{\langle k,v\rangle:v\in\T_k\}$ the set $\mathcal{L}_{ik}$.  We denote by
$\mathcal{L}$ the set of all backhaul and access links in the network.

\vspace{-0.4cm}
\subsection{Load-Coupling Model}
Let $r_{ij}$ be the bit rate demand on the link $\langle
i,j\rangle$. Denote by $\gamma_{ij}$ the signal-to-interference-and-noise ratio
(SINR) from $i$ to $j$.  Without loss of generality, we use an (RU)
to refer to the minimum unit for resource allocation. The bandwidth per RU is
denoted by $B$. In the denominator in~\eqref{eq:load}, $B\log_2(1+\gamma_{ij})$
computes the achievable bit rate per RU\@. We assume that there are $M$ RUs in
total, such that $MB\log_2(1+\gamma_{ij})$ is the total achievable bit rate for
UE $j$.  In~\eqref{eq:load}, $x_{ij}$ is then defined to be the proportion of
RUs used by the transmission link $\langle i,j \rangle$, among all RUs in
cell $i$.  The sum of the proportion of allocated RUs in any cell $i$, i.e.,
$\sum_{j\in\T_i}x_{ij}$, is
defined to be the \emph{load} of cell $i$, which is bounded by the full load,
i.e. $\sum_{j\in\mathcal{T}_i}x_{ij}\leq 1~i\in\mathcal{B}\cup\mathcal{R}$.
\vskip -10pt
\begin{equation} 
    x_{ij}=\frac{r_{ij}}{MB\log_2(1+\gamma_{ij})} 
\label{eq:load}
\end{equation}
\vskip -5pt

The SINR on any RU allocated to $\langle i,j\rangle$ is given
by~\eqref{eq:sinr}. In the nominator, $p_{ij}$ is the transmission power of an RU of link $\langle i,j\rangle$ in cell $i$. The value of $g_{ij}$ is
the power gain from $i$ to $j$. In the denominator
of~\eqref{eq:sinr}, recall that $x_{vu}$ represents the proportion of occupied
RU by $\langle v,u\rangle$ in cell $v$. The
value of $x_{vu}$ is then interpreted as the likelihood that $\langle i,j\rangle$ receives
interference from $\langle v,u\rangle$ on the RU\@. Note that $\langle
v,u\rangle\in\mathcal{L}\backslash\mathcal{L}_{ij}$, which is the set of all links that
are not required to be orthogonal to $\langle i,j\rangle$.
\vskip -5pt
\begin{equation}
\gamma_{ij}=\frac{p_{ij}g_{ij}}{\sum_{\langle v,u\rangle\in
\mathcal{L}\backslash\mathcal{L}_{ij}}p_{vu}g_{vj}x_{vu}+\sigma^2}
\label{eq:sinr}
\end{equation}
\vskip -5pt

By~\eqref{eq:load} and~\eqref{eq:sinr}, one can observe that a change on
$x_{uv}$ for any link $\langle u,v\rangle$ may cause a variation after the SINR
of some link $\langle i,j\rangle$, thus leading to a new value of $x_{ij}$,
i.e., the required resource consumption for link $\langle i,j\rangle$. 
Thus, the levels of resource consumption are inherently coupled. This
relationship, as characterized by~\eqref{eq:load} and~\eqref{eq:sinr}, is called
load-coupling.



\vspace{-0.35cm}
\subsection{Computation of Transmission Energy}

Recall that $x_{ij}$ represents the proportion of consumed RU of link
$\langle i,j \rangle$. Hence, the number of RUs that are used for transmission
by $\langle i,j\rangle$ is $Mx_{ij}$. On each RU, the transmit power is
$p_{ij}$. Then the energy consumption on link $\langle i,j\rangle$ is
$Mp_{ij}x_{ij}$. 
We now focus on how to compute $x_{ij}$ in the load-coupling model
in~\eqref{eq:load} and~\eqref{eq:sinr}. Let $n=n_{\T}+n_{\R}$.  The proportions
of RU consumption for all potential links in the network are represented by the
vector $\bm{x}=[~{[x_{11},\dots,x_{|\C_1|1}]},\dots,{[x_{1n},
\dots,x_{|\C_n|n}]}~]$.

By plugging~\eqref{eq:sinr}
in~\eqref{eq:load}, we get the function of the proportion of consumed RUs by
$\langle i,j\rangle$ in~\eqref{eq:xfunc} below. For vector
$\bm{x}$ satisfying the cell-load coupling relation in the system model 
$x_{ij}=F_{ij}(\bm{x})$ holds for all $\langle i,j\rangle\in\mathcal{L}$. 
\begin{equation}
F_{ij}(\bm{x})=
\frac{r_{ij}}
{MB\log_2(1+
\frac{p_{ij}g_{ij}}{\sum\limits_{\langle
    v,u\rangle\in\mathcal{L}\backslash\mathcal{L}_{ij}}p_{uv}g_{uj}x_{uv}+\sigma^2}
)}
\label{eq:xfunc}
\end{equation}
\vskip -5pt

It can be verified by observing the concavity of function $F_{ij}(\bm{x})$ that
$F_{ij}$ is a standard interference function (SIF) in respect of $\bm{x}$
\cite{Cavalcante:2014jd,Yates:1995eh}.  An SIF has the following property:
starting from an arbitrary positive $\bm{x}^{(0)}$, if the fixed point of
function $F_{ij}$ exists, then it is unique, can be iteratively computed by
$x^{(k)}_{ij}=F_{ij}(\bm{x}^{(k-1)})$ $(k\geq 1)$. 

\vspace{-0.3cm}
\section{Problem Formulation}
For any UE $j\in\T$, we use a variable $a_j$ to indicate the UE's serving cell,
i.e. $a_j=i$ if UE $j$ is currently served by cell $i$. Similarly, for any
RC $k\in\R$, we use $a_k=i$ to indicate that RC $k$ is connected to MC $i$ with a
wireless backhaul. For any $j$, $a_j\in\mathcal{C}_j$
for all $j\in\T\cup\R$. The vector $\bm{a}$ then denotes
the association among MCs, RCs and UEs.
\begin{figure}[!h]
\vskip -20pt
\begin{subequations}
\begin{alignat}{2}
[\textit{MinE}]\quad &
\min\limits_{\bm{x},\bm{a},\bm{r}}\quad M\sum_{j=1}^{n}p_{a_{j}j}x_{a_{j}j}
\label{eq:obj}\\
s.t. \quad & \bm{x}=\bm{F}(\bm{x},\bm{a},\bm{r})\label{eq:cstr_lc} \\
& r_{a_j j}=d_j \quad j\in\T \label{eq:cstr_rj}\\
& r_{a_k k}=\sum_{j\in\T_k}d_j \quad k\in\R \label{eq:cstr_rk}\\
& \sum_{k\in\mathcal{R}:a_k=i}x_{ik} + \sum_{j\in\mathcal{T}:a_j=i}x_{ij}\leq 1 \quad
i\in\mathcal{B} \label{eq:cstr_ld}\\
& x_{a_{k}k} + \sum_{j\in\mathcal{T}:a_j=k}x_{kj} \leq 1  \quad
k\in\mathcal{R}\label{eq:cstr_ld2} \\ 
& a_j\in\mathcal{C}_j \quad j\in\mathcal{T}\cup\mathcal{R} \label{eq:cstr_C}
\end{alignat}
\label{eq:mine}
\vskip -20pt
\end{subequations}
\end{figure}
The energy-aware relay selection problem, a.k.a. \textit{MinE}, is formulated in
\eqref{eq:mine}. The objective of minimizing the energy on all links is given
in~\eqref{eq:obj}. 
Constraint~\eqref{eq:cstr_lc} ensures that $\bm{x}$ satisfies
the coupling relationship in the system model. Constraint~\eqref{eq:cstr_rj}
guarantees that the bit rate demand of any UE $j\in\T$ is satisfied.
Constraint~\eqref{eq:cstr_rk} ensures sufficient bit rate on each backhaul link. 
Constraint~\eqref{eq:cstr_ld}
and~\eqref{eq:cstr_ld2} are imposed to limit the proportion of consumed RUs in
each cell to be no more than 1, corresponding to the full load constraint for
MCs and RCs, respectively. Constraint~\eqref{eq:cstr_C} is imposed such
that the selected cell for a backhaul/access link is within the candidate set. 

\vspace{-0.3cm}
\section{Complexity Analysis}

\begin{theorem}
    \textit{MinE} is $\mathcal{NP}$-hard.
\label{thm:np-hard}
\end{theorem}
\begin{proof}
We reduce the Maximum Independent Set (MIS) problem to \textit{MinE}. We
construct a specific HetNet scenario. For each UE, there is one potential MC and
one potential RC as candidate serving cells. Correspondingly, for any undirected
graph instance $\G$ with $N$ nodes ($N\geq 2$) in the MIS problem, we define $N$
UEs. Thus, for any node in $\G$, we have one UE, one MC and one RC\@. We use $1,
2, \ldots, N$ to index the nodes in graph $\G$. We use the term
``\textit{neighboring}'' to refer to the relationship of any two entities that
are associated respectively to two neighboring nodes in $\G$.

For any node $i$ in $\G$, we set the gain from MC $i$ to UE $i$ to $1.0$, the
gain from RC $i$ to UE $i$ to $6.0$, and the gain from MC $i$ to RC $i$ to
$3.0$, respectively.  For any two neighbouring nodes $i$ and $k$ (meaning that
there is an edge between node $i$ and node $k$) in graph $\G$, we set the gain
from RC $i$ to UE $k$ to a small positive value $\epsilon$. The gain values
other than the above three cases are set to be negligible, treated as zero. The
noise $\sigma^2$ is set to $1.0$.  The values of gain and noise can be scaled
without affecting the validity of the proof.  The transmit powers of MCs and RCs
are set to $1.0$ and $0.5$, respectively.The bit rate demand for any UE is set
to $1.0$.  

Due to space limit, we give a sketch of the line of arguments. For any UE $i$, if we have it
served by RC $i$, then all the resource in RC $i$ is in use. If any RC $k$
neighboring to UE $i$ is activated, then UE $i$ would receive the interference
from RC $k$, leading to that UE $i$'s demand cannot be satisfied anymore by the
access link from RC $i$ to UE $i$. Hence in a feasible solution, any pair of two
neighboring UEs cannot be simultaneously served by their
corresponding RCs. In addition, one can verify that it is always better to serve
any UE $i$ with RC $i$ rather than MC $i$ for energy saving. Thus we finish the
reduction by concluding
that, to solve this constructed problem
instance is to maximize the number of activated RCs, subject to that at most one
RC of any pair of neighboring RCs can be in use. Hence the conclusion.
\end{proof}

\vspace{-0.45cm}
\section{Energy Minimization via Optimality Condition}
\label{sec:em}


This section aims to seek for an effective strategy to deal with the
combinatorial nature of \textit{MinE}. A partial optimality condition is derived
below, based on which we propose the relay selection algorithm.

\vspace{-0.35cm}
\subsection{Optimality Condition}
\label{subsec:opt}
We introduce some notations for deriving the optimality condition.  Denote by
$\hat{\bm{a}}$ and $\check{\bm{a}}$ any two associations, such that $\exists j$
$\hat{a}_j\neq\check{a}_j$.  Denote
$\mathpzc{l}=\{j:\hat{a}_j\neq\check{a}_j,~j\in\T\cup\R\}$.  Denote by
$\hat{\bm{x}}$ and $\check{\bm{x}}$ the fixed points of the function $\bm{F}$
under $\hat{\bm{a}}$ and $\check{\bm{a}}$, respectively.  Denote by $\hat{e}$
the total transmission energy with association $\hat{\bm{a}}$, i.e.
$\hat{e}=\sum_{j=1}^{n}p_{\hat{a}_j j}\hat{x}_{\hat{a}_j j}$, and by $\check{e}$
the total transmission energy with association $\check{\bm{a}}$, i.e.
$\check{e}=\sum_{j=1}^{n}p_{\check{a}_j j}\check{x}_{\check{a}_j j}$. 
\begin{definition}
    We define the following function for any $i\in\C_j$,
    and any $j\in\T\cup\R$, where
    $\mathpzc{t}$ is a non-empty subset of $\T\cup\R$.
\begin{equation}
G_{ij}(\bm{x},\bm{a},\mathpzc{t})=\left\{
    \begin{array}{ll}
        F_{ij}(\bm{x},\bm{a}) & j\in\mathpzc{t} \\
        x_{ij} & \textnormal{otherwise}
    \end{array}
\right.
\label{eq:F}
\end{equation}
\end{definition}
\begin{theorem}
    (\textbf{Optimality Condition}) $\hat{e}<\check{e}$ if and only if for some set
    $\mathpzc{t}$
    $(\mathpzc{l}\subseteq\mathpzc{t}\subseteq\T\cup\R)$ such
    that:
\begin{enumerate}
    \item $\sum_{j\in\mathpzc{t}}p_{\hat{a}_j j}x^{\mathpzc{t}}_{\hat{a}_j
    j}<\sum_{j\in\mathpzc{t}}p_{\check{a}_j j}\check{x}_{\check{a}_j j}$ where
    any $x^{\mathpzc{t}}_{\hat{a}_j j}$ with $j\in\T\cup\R$ is an
    element of $\bm{x}^{\mathpzc{t}}$, and $\bm{x}^{\mathpzc{t}}$ is the
    fixed point of $\bm{G}(\bm{x},\hat{\bm{a}}, \mathpzc{t})$,
    with $\check{\bm{x}}$ being the starting point.
\item $F_{\hat{a}_j j}(\bm{x}^{\mathpzc{t}},\hat{\bm{a}})\leq
    \check{x}_{\check{a}_j j}$ for any
    $j\notin\mathpzc{t}$.
\end{enumerate}
\label{thm:opt}
\end{theorem}
\begin{proof}
\vskip -7pt
The necessity can be proved straightforwardly by letting
$\mathpzc{t}=\mathcal{T}\cup\mathcal{R}$. For the sufficiency, 
the basic idea is to prove $\hat{\bm{x}}\leq\bm{x}^{\mathpzc{t}}$ by using
Condition 2), and then combine it with Condition 1) to compute respectively
$\hat{e}$ and $\check{e}$. 
Suppose there
exists some set $\mathpzc{t}$
($\mathpzc{l}\subseteq\mathpzc{t}\subseteq\T\cup\R$) satisfying 1) and 2).
We consider the fixed-point iterations
$\hat{\bm{x}}^{(k)}=\bm{F}(\hat{\bm{x}}^{(k-1)},\hat{\bm{a}})$ $(k\geq 0)$. Let
$\hat{\bm{x}}^{(0)}=\bm{x}^{\mathpzc{t}}$.  For $j\in\mathpzc{t}$, since
$x^{\mathpzc{t}}_{\hat{a}_j j}=G_{\hat{a}_j
j}(\bm{x}^{\mathpzc{t}},\hat{\bm{a}},\mathpzc{t})=F_{\hat{a}_j
j}(\bm{x}^{\mathpzc{t}},\hat{\bm{a}})$, combined with the construction that
$\hat{\bm{x}}^{(0)}=\bm{x}^{\mathpzc{t}}$, we have $\hat{x}_{\hat{a}_j
j}^{(1)}=F_{\hat{a}_j j}(\hat{\bm{x}}^{(0)},\hat{\bm{a}})=G_{\hat{a}_j
j}(\hat{\bm{x}}^{(0)},\hat{\bm{a}},\mathpzc{t})=G_{\hat{a}_j
j}(\bm{x}^{\mathpzc{t}},\hat{\bm{a}},\mathpzc{t})=x^{\mathpzc{t}}_{\hat{a}_j
j}=\hat{x}^{(0)}_{\hat{a}_j j}$. For $j\notin\mathpzc{t}$,  we have
$\hat{x}_{\hat{a}_j j}^{(1)}=F_{\hat{a}_j j}(\hat{\bm{x}}^{(0)},\hat{\bm{a}})$.
By the construction that $\hat{\bm{x}}^{(0)}=\bm{x}^{\mathpzc{t}}$ and condition
2), $\hat{x}_{\hat{a}_j j}^{(1)}=F_{\hat{a}_j
j}(\hat{\bm{x}}^{(0)},\hat{\bm{a}})=F_{\hat{a}_j
j}(\bm{x}^{\mathpzc{t}},\hat{\bm{a}})\leq x_{\hat{a}_j
j}^{\mathpzc{t}}=\hat{x}_{\hat{a}_j j}^{(0)}$ holds. Therefore, we have 

\vskip -7pt
\begin{equation}
\hat{\bm{x}}^{(1)}\leq\hat{\bm{x}}^{(0)}=\bm{x}^{\mathpzc{t}}.
\end{equation}
\vskip -7pt

We first consider $\hat{e}$. By the monotonicity of $\bm{F}(\bm{x},\hat{\bm{a}})$ in $\bm{x}$, we have the
following property. For any $k\geq 0$, if $\bm{x}^{(k)}\leq\bm{x}^{(k-1)}$, then
$\bm{F}(\bm{x}^{(k)},\hat{\bm{a}})\leq\bm{F}(\bm{x}^{(k-1)},\hat{\bm{a}})$
holds, which would directly lead to $\bm{x}^{(k+1)}\leq\bm{x}^{(k)}$.  According
to the discussion above, we have $\bm{x}^{(k+1)}\leq\bm{x}^{(k)}$ for
$k=0$. We therefore conclude by mathematical induction that
$\hat{\bm{x}}\leq\cdots\leq\hat{\bm{x}}^{(1)}\leq\hat{\bm{x}}^{(0)}=\bm{x}^{\mathpzc{t}}$. Thus,
we have 

\vskip -7pt
\begin{equation}
\hat{e}\leq\sum_{j=1}^{n}p_{\hat{a}_j j}x^{\mathpzc{t}}_{\hat{a}_j j}
\end{equation}
\vskip -7pt

We then consider $\check{e}$.  For any $j\notin\mathpzc{t}$, we have
$x_{\hat{a}_j j}^{\mathpzc{t}}=G_{\hat{a}_j
j}(\bm{x}^{\mathpzc{t}},\hat{\bm{a}},\mathpzc{t})$. Since
$\mathpzc{l}\subseteq\mathpzc{t}$, we conclude $j\notin\mathpzc{l}$ for any
$j\notin\mathpzc{t}$. Therefore, according to the definition of 
$\mathpzc{l}$, we have $\hat{a}_j=\check{a}_j$ for $j\notin\mathpzc{t}$.
Note that in condition 1), $\check{\bm{x}}$ is the starting
point for the fixed-point iterations of the function
$\bm{G}(\bm{x},\check{\bm{a}},\mathpzc{t})$. According to the definition of the
function $G_{ij}$ in~\eqref{eq:F}, for any $j\notin\mathpzc{t}$, we have
$x^{\mathpzc{t}}_{\hat{a}_j j}=\check{x}_{\hat{a}_j j}=\check{x}_{\check{a}_j
j}$. Thus we conclude $\sum_{j\notin\mathpzc{t}}p_{\hat{a}_j
j}x^{\mathpzc{t}}_{\hat{a}_j j}=\sum_{j\notin\mathpzc{t}}p_{\check{a}_j
j}x_{\check{a}_j j}$.  For any $j\in\mathpzc{t}$, condition 1) shows that
$\sum_{j\in\mathpzc{t}}p_{\hat{a}_j j}x^{\mathpzc{t}}_{\hat{a}_j
j}<\sum_{j\in\mathpzc{t}}p_{\check{a}_j j}\check{x}_{\check{a}_j j}$.
Therefore, we obtain 

\vskip -15pt
\begin{equation}
\check{e}>\sum_{j=1}^{n}p_{\hat{a}_j j}x^{\mathpzc{t}}_{\hat{a}_j j}
\end{equation}
\vskip -5pt

Hence the conclusion $\hat{e}<\check{e}$.
\end{proof}
Given any two associations $\hat{\bm{a}}$ and $\check{\bm{a}}$,
Theorem~\ref{thm:opt} serves as a sufficient and necessary condition for
checking whether $\hat{\bm{a}}$ leads to a better energy performance than
$\check{\bm{a}}$.  When $\mathpzc{t}\subset\T\cup\R$, it is called asynchronous
fixed-point iterations~\cite{Yates:1995eh}, which is rather faster if
$|\mathpzc{t}|\ll |\T\cup\R|$ and can be implemented in a distributed manner. 

\vspace{-0.3cm}
\subsection{Algorithm Design}
\label{subsec:alg}

We use the function $A_j(\bm{x})= \argmin\limits_{i\in\C_j}p_{ij}x_{ij}$ to assign each $j\in\T\cup\R$ to
the cell with lowest energy for transmitting to $j$. 
Algorithm~\ref{alg:rs} below takes an
initial association $\check{\bm{a}}$ as the input, and outputs the optimized
association $\hat{\bm{a}}$. The pre-defined parameter $\eta$ indicates
the maximum number of loop rounds. The vectors $\bm{a}$ and $\bm{x}$ are
iteratively updated by functions $\bm{A}$ and $\bm{F}$, respectively.  Once
$\bm{a}^{(k)}=\bm{a}^{(k-1)}$ holds for any round $k$, meaning that there is no
update on vector $\bm{a}^{(k)}$, the algorithm terminates and returns
the optimized association $\hat{\bm{a}}$.  In each round, the set
$\mathpzc{l}^{(k)}$ records the positions of all the different elements between
$\bm{a}^{(k)}$ and $\check{\bm{a}}$. In Lines~\ref{l:t1} and~\ref{l:t2}, the
asynchronous fixed-point iterations are applied with respect to set
$\mathpzc{t}$, with $\mathpzc{t}\supseteq\mathpzc{l}^{(k)}$. Lines~\ref{l:opt0}
and~\ref{l:opt} use the optimality condition to check if the new association
$\bm{a}^{(k)}$
would improve the transmission energy, with numerical tolerances $\epsilon_1$
and $\epsilon_2$.
\vspace{-0.5cm}
\begin{algorithm}[!h]
$\bm{a}^{(0)}\leftarrow \check{\bm{a}}$, $\hat{\bm{a}}\leftarrow\check{\bm{a}}$\;
$\bm{x}^{(0)}\leftarrow$ fixed point of $\bm{F}(\bm{x},\bm{a}^{(0)})$\;
\For{$k\leftarrow 1~\textnormal{\textbf{to}}~\eta$\label{l:lpstart}}{
$\bm{a}^{(k)}\leftarrow\bm{A}(\bm{x}^{(k-1)})$\;
\If{$\bm{a}^{(k)}=\bm{a}^{(k-1)}$}{\textbf{break}\;}
$\bm{x}^{(k)}\leftarrow\bm{F}(\bm{x}^{(k-1)},\bm{a}^{(k)})$\;
$\mathpzc{l}^{(k)}\leftarrow\{j:a^{(k)}_{j}\neq\check{a}_{j},~j\in\T\cup\R\}$\;
choose a set $\mathpzc{t}$ such that $\mathpzc{l}^{(k)}\subseteq\mathpzc{t}$\label{l:t1}\;
$\bm{x}^{\mathpzc{t}}\leftarrow$ fixed point of
$\bm{G}(\bm{x},\hat{\bm{a}},\mathpzc{t})$, starting from
$\bm{x}^{(k)}$\label{l:t2}\;
\If{$\sum_{j\in\mathpzc{t}}p_{a_j^{(k)}j}x^{\mathpzc{t}}_{a_j^{(k)}
j}<\sum_{j\in\mathpzc{t}}p_{\check{a}_j j}\check{x}_{\check{a}_j
j}+\epsilon_1$\label{l:opt0}  \\ \quad\quad\quad
$\wedge$ $F_{a_j^{(k)} j}(\bm{x}^{\mathpzc{t}},\hat{\bm{a}})\leq
\check{x}_{\check{a}_j j}+\epsilon_2$\label{l:opt}}
{
    $\hat{\bm{a}}\leftarrow\bm{a}^{(k)}$\label{l:lpend}
}
}
\KwRet{$\hat{\bm{a}}$}\;
\caption{Relay Selection} 
\label{alg:rs}
\end{algorithm}

\vspace{-0.7cm} 
\section{Numerical and Simulation Results}
\vspace{-0.1cm}

For simulation, $7$ MCs are deployed at the center of a hexagonal region, each
with the distance of $500$ meters to its neighbor MC\@. In each hexagonal
region, $2$ or $4$ RCs as well as $20$ UEs are randomly placed. The network
operates at $2$ GHz. Each RU is set to $180$ kHz bandwidth and the bandwidth for
each cell is $20$ MHz. The noise power spectral density is set to $-174$ dBm/Hz.
We remark that the simulation settings follow the 3GPP standardization
document~\cite{release13-1}, to be consistent with expected 5G network scenarios
in terms of bandwidth and network density. Also, the path loss of MCs and RCs
follow the standard 3GPP urban macro and micro models~\cite{release9},
respectively. With the simulation setup and based on one thousand instances, the
peak rate can achieve 1 Gbps; this is consistent with~\cite{release13-1} and the
expectation of 5G. The average rate, which is naturally lower than the peak one,
depends on user density and resource sharing The shadowing coefficients are
generated by the log-normal distribution with $6$ dB and $3$ dB standard
deviation~\cite{release9}, for MCs and RCs, respectively.  The maximum transmit
power levels for MCs and RCs are set to $800$ mW and $50$ mW per RU,
respectively. Simulations are run over multiple data sets and averaged
afterwards.

\begin{figure}[!h]
    \centering
    \vskip -15pt
    \includegraphics[width=0.8\linewidth]{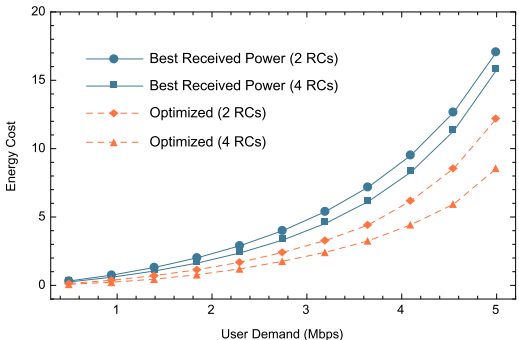}
     \vskip -15pt
    \caption{User demand versus energy cost.}
\label{fig:sim}
\vspace{-0.35cm}
\end{figure}

In \figurename~\ref{fig:sim}, there are two network scenarios, in which 2 and 4
RCs are deployed in each hexagonal region, respectively. As references for
comparison, each UE (or RC) $j$ is associated with the cell in $\C_j$ with the
best received power. As expected, the 4-RCs case benefits more on energy
performance via relay selection, compared to the 2-RCs case, since that each UE
has more options for choosing its serving cell in the former. For these two
cases, the improvements by using the proposed algorithm are $34\%$ and $47\%$,
respectively. Furthermore, the improvement becomes larger, with the increase of
the user demand, which indicates that an appropriate relay selection is crucial
for a network with heavy data traffic. We remark that for the best-received
power based relay selection, the network can still also benefit from deploying more
RCs on energy cost. In other words, the energy cost can be reduced by deploying
more RCs, without optimizing the relay selection. However, one can see from the
numerical results that the corresponding gain is far less compared to optimizing the relay
selection. 

\vspace{-0.2cm}
\section{Conclusion}

This paper has provided insights as well as an algorithm 
for energy-aware relay selection in load-coupled OFDMA cellular networks.
The algorithm exhibits good performance for energy saving.

\vspace{-0.5cm}
\section*{Acknowledgement}
This work has been supported by the Swedish Research Council and the
Link\"{o}ping-Lund Excellence Center in Information Technology (ELLIIT), Sweden,
and the European Union Marie Curie project MESH-WISE (FP7-PEOPLE-2012-IAPP\@:
324515), DECADE (H2020-MSCA-2014-RISE\@: 645705), and WINDOW
(FP7-MSCA-2012-RISE\@: 318992). The work of D. Yuan has been
carried out within European FP7 Marie Curie IOF project 329313.

\vspace{-0.2cm}
\bibliographystyle{IEEEtran}
\bibliography{ref}
\end{document}